\documentclass[a4paper,11pt,twocolumn]{article}
\usepackage{amssymb}
\usepackage{amsthm}
\usepackage[utf8]{inputenc}
\usepackage{array}
\usepackage[T1]{fontenc}
\usepackage{lmodern}
\usepackage{pgf,tikz}
\usepackage{enumitem}
\setlist{noitemsep,leftmargin=5.5mm}
\usetikzlibrary{arrows}
\usepackage{vmargin}
\setmarginsrb{2.5cm}{2cm}{2.5cm}{1cm}{0cm}{0cm}{0cm}{1cm}
\usepackage{caption}
\usepackage{subcaption}

\title{A lower bound on the order of the largest induced linear forest in
  triangle-free planar graphs}

\usepackage{hyperref}
\usepackage[affil-it]{authblk}

\author[a]{François Dross}
\author[a]{Mickael Montassier}
\author[b]{Alexandre Pinlou}

\affil[a]{{\small Université
    de Montpellier, LIRMM}} \affil[b]{{\small Université Paul-Valéry Montpellier 3,
    LIRMM\medskip}} \affil[ ]{{\small 161 rue Ada, 34095 Montpellier
    Cedex 5, France}} \affil[
]{\href{mailto:francois.dross@lirmm.fr,mickael.montassier@lirmm.fr,alexandre.pinlou@lirmm.fr}{\small{\{francois.dross,mickael.montassier,alexandre.pinlou\}@lirmm.fr}}}
\date{}

\begin{document}

\maketitle
\newtheorem{theo}{Theorem}
\newtheorem{cor}[theo]{Corollary}
\newtheorem{lemm}[theo]{Lemma}
\newtheorem{prop}[theo]{Property}
\newtheorem{obs}[theo]{Observation}
\newtheorem{conj}[theo]{Conjecture}
\newtheorem{claim}[theo]{Claim}

\begin{abstract}
We prove that every triangle-free planar graph of order $n$ and size $m$ has an induced linear forest with at least $\frac{9n - 2m}{11}$ vertices, and thus at least $\frac{5n + 8}{11}$ vertices. Furthermore, we show that there are triangle-free planar graphs on $n$ vertices whose largest induced linear forest has order $\lceil \frac{n}{2} \rceil + 1$.
\end{abstract}

\section{Introduction}
In this article, we only consider simple finite graphs. All considered planar graphs are supposed to be embedded in the plane.

We look into the problem of finding large induced forests in planar graphs. Albertson and Berman~\cite{AlbertsonBerman} conjectured that every planar graph admits an induced forest on at least half of its vertices. This conjecture, if true, would be tight, as shown by the disjoint union of copies of the complete graph on four vertices. One of the motivations of this conjecture is that it would imply that every planar graph admits an independent set on at least one fourth of its vertices, the only known proof of which relies on the Four Colour Theorem. However, this conjecture appears to be very hard to prove. The best known result for planar graphs is that every planar graph admits an induced forest on at least two fifths of its vertices. This is a consequence of the theorem of $5$-acyclic colourability of planar graphs of Borodin~\cite{Borodin}. 

The conjecture of Albertson and Berman has been proved and strengthened for smaller classes of graphs. For example, Hosono~\cite{Hosono} showed that every outerplanar graph admits an induced forest on at least two thirds of its vertices, which is tight. Akiyama and Watanabe~\cite{Akiyama}, and Albertson and Rhaas~\cite{Albertson} independently conjectured that every bipartite planar graph admits an induced forest on at least five eighths of its vertices, which is tight. For triangle-free planar graphs (and thus in particular for bipartite planar graphs), it is proven that every triangle-free planar graph of order $n$ and size $m$ admits an induced forest of order at least $(38n-7m)/44$, and thus at least $(6n+7)/11$~\cite{girth4-5}.

An interesting variant of this problem is to look for large induced forests with bounded maximum degree. A forest with maximum degree $2$ is called a \emph{linear forest}. 

The problem for linear forests was solved for outerplanar graphs by Pelsmajer~\cite{Pelsmajer}: every outerplanar graph admits an induced linear forest on at least four sevenths of its vertices, and this is tight. More generally, the problem for a forest of maximum degree at most $d$, with $d \ge 2$, was solved for graphs with treewidth at most $k$ for all $k$ by Chappel and Pelsmajer~\cite{chappell2013maximum}. Their result in particular extends the results of Hosono and Pelsmayer on outerplanar graphs to series-parallel graphs, and generalises it to graphs of bounded treewidth.

In this paper we focus on linear forests. Chappel conjectured that every planar graph admits an induced linear forest on at least four ninths of its vertices. Again, this would be tight if true. Poh~\cite{Poh} proved that every planar graph can have its vertices be partitioned into three sets, each inducing a linear forest, and thus that every planar graph admits an induced linear forest on at least one third of its vertices. In this paper, we prove and strengthen Chappel's conjecture for a smaller class of graphs, the class of triangle-free planar graphs. Observe that planar graphs with arbitrarily large girth can have an arbitrarily large treewidth, so in this setting the best result known to date is that every triangle-free planar graph admits an induced linear forest on at least one third of its vertices.

We prove the following theorem:

\begin{theo}\label{main}
Every triangle-free planar graph of order $n$ and size $m$ admits an induced linear forest of order at least $\frac{9n - 2m}{11}$.
\end{theo}

Thanks to Euler's formula, we can derive the following corollary:

\begin{cor}
Every triangle-free planar graph of order $n$ admits an induced linear forest of order at least $\frac{5n + 8}{11}$.
\end{cor}

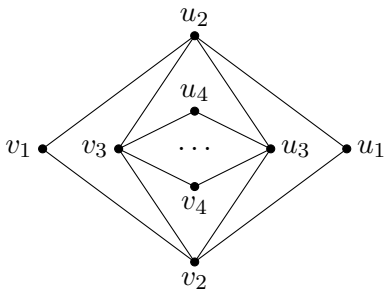
\begin{figure}
\begin{center}
\begin{tikzpicture}
\coordinate (u1) at (2,0) {};
\coordinate (v1) at (-2,0) {};
\coordinate (u2) at (0,1.5) {};
\coordinate (v2) at (0,-1.5) {};
\coordinate (u3) at (1,0) {};
\coordinate (v3) at (-1,0) {};
\coordinate (u4) at (0,0.5) {};
\coordinate (v4) at (0,-0.5) {};

\draw (u1) node [right] {$u_1$};
\draw (v1) node [left] {$v_1$};
\draw (u2) node [above] {$u_2$};
\draw (v2) node [below] {$v_2$};
\draw (u3) node [right] {$u_3$};
\draw (v3) node [left] {$v_3$};
\draw (u4) node [above] {$u_4$};
\draw (v4) node [below] {$v_4$};
\draw (0,0) node {$\dots$};

\draw (u1) -- (u2);
\draw (u1) -- (v2);
\draw (v1) -- (u2);
\draw (v1) -- (v2);
\draw (u2) -- (u3);
\draw (u2) -- (v3);
\draw (v2) -- (u3);
\draw (v2) -- (v3);
\draw (u3) -- (u4);
\draw (u3) -- (v4);
\draw (v3) -- (u4);
\draw (v3) -- (v4);

\draw [fill=black] (u1) circle (1.5pt) ;
\draw [fill=black] (v1) circle (1.5pt) ;
\draw [fill=black] (u2) circle (1.5pt) ;
\draw [fill=black] (v2) circle (1.5pt) ;
\draw [fill=black] (u3) circle (1.5pt) ;
\draw [fill=black] (v3) circle (1.5pt) ;
\draw [fill=black] (u4) circle (1.5pt) ;
\draw [fill=black] (v4) circle (1.5pt) ;

\end{tikzpicture}
\end{center}
\caption{The graph $G_k$ of Claim~\ref{lb}.}\label{figlb}
\end{figure}

For a graph $G = (V,E)$, and $S \subset V$, let $G[S]$ denote the subgraph of $G$ induced by $S$.
Note that we cannot hope to get a better lower bound than $\frac{n}{2} + 1$. Indeed, we prove the following claim:

\begin{claim}\label{lb}
For all integer $n \ge 2$, there exists a triangle-free planar graph of order $n$ whose largest induced linear forest has order $\lceil \frac{n}{2} \rceil + 1$.
\end{claim}

\begin{proof}
Let us build such a graph for $n = 2k$. For odd $n$, adding an isolated vertex to the graph of order $n-1$ yields the result.

Let $G_k$ be defined by $G_k = (\bigcup_{1 \le i \le k} \{u_i, v_i\}, \bigcup_{1 \le i \le k - 1}$ $ \{u_iu_{i+1}, u_iv_{i+1}, v_iu_{i+1}, v_iv_{i+1}\})$, as represented in Figure~\ref{figlb}. Let us prove by induction on $k \ge 1$ that the largest induced linear forest of $G_k$ has order $k + 1$.

\begin{itemize}
\item
For $k = 1$, $G_1$ is the graph with two vertices and no edge, and $G_1$ is its own largest induced linear forest, with order $2 = k+1$.

\item
For $k = 2$, $G_2$ is a cycle of length $4$, any three vertices of $G_2$ induce a linear forest of order $3 = k+1$, and $G_2$ is not a linear forest (thus it has no induced linear forest of order $4$).

\item
For $k = 3$, $\{u_1, v_1, u_3, v_3\}$ induces a linear forest in $G_3$, and it is easy to check that no five vertices of $G_3$ induce a linear forest.

\item
Suppose $k \ge 4$. By induction hypothesis, $G_{k-1}$, $G_{k-2}$, and $G_{k-3}$ have a largest induced forest of order $k$, $k-1$, and $k-2$ respectively. Adding $u_k$ and $v_k$ to any induced linear forest of $G_{k-2}$ leads to an induced linear forest of $G_k$, thus $G_k$ has an induced linear forest of order $k+1$. All that remains to prove is that $G_k$ has no induced linear forest of order at least $k+2$.

Let $F \subset V(G_k)$ be a set inducing a linear forest of $G_k$. Let us prove that $|F| \le k+1$. As $F \backslash \{u_k,v_k\}$ induces a linear forest in $G_{k-1}$, we have $|F \backslash \{u_k,v_k\}| \le k$. Similarly, $|F \backslash \{u_k,v_k,u_{k-1},v_{k-1}\}| \le k-1$ and $|F \backslash \{u_k,v_k,u_{k-1},v_{k-1},u_{k-2},v_{k-2}\}| \le k-2$. If $|\{u_k,v_k\} \cap F| \le 1$, then $|F| = |\{u_k,v_k\} \cap F| + |F \backslash \{u_k,v_k\}| \le k + 1$. Suppose now that $|\{u_k,v_k\} \cap F| > 1$, i.e. $\{u_k,v_k\} \subset F$. At most one of $u_{k-1}$ and $v_{k-1}$ is in $F$, otherwise $G_k[F]$  has a cycle. If $\{u_{k-1},v_{k-1}\} \cap F = \emptyset$, then $|F| = |\{u_k,v_k\} \cap F| + |\{u_{k-1},v_{k-1}\} \cap F| + |F \backslash \{u_k,v_k,u_{k-1},v_{k-1}\}| \le 2 + k - 1 = k+1$. Assume now that $\{u_{k-1},v_{k-1}\} \cap F \ne \emptyset$, w.l.o.g. $u_{k-1} \in F$. We have $\{u_{k-2},v_{k-2}\} \cap F = \emptyset$, otherwise $u_{k-1}$ has degree at least $3$ in $G_k[F]$. Hence, $|F| = |\{u_k,v_k\} \cap F| + |\{u_{k-1},v_{k-1}\} \cap F| + |\{u_{k-2},v_{k-2}\} \cap F| + |F \backslash \{u_k,v_k,u_{k-1},v_{k-1},u_{k-2},v_{k-2}\}| \le 2 + 1 + k - 2 = k+1$.
\end{itemize}
\vspace{-1cm}
\end{proof}

\section{Proof of Theorem~\ref{main}}
Consider a graph $G = (V,E)$. For a set $S \subset V$, let $G - S$ be the graph obtained from $G$ by removing the vertices of $S$ and all the edges incident to a vertex of $S$. If $x \in V$, then we denote $G - \{x\}$ by $G - x$. For a set $S$ of vertices such that $S \cap V = \emptyset$, let $G + S$ denote the graph obtained from $G$ by adding the vertices of $S$. If $x \notin V$, then we denote $G + \{x\}$ by $G + x$. For a set $F$ of pairs of vertices of $G$ such that $F \cap E = \emptyset$, let $G + F$ be the graph constructed from $G$ by adding the edges of $F$. If $e$ is a pair of vertices of $G$ and $e \notin E$, then we denote $G + \{e\}$ by $G + e$. If $x \in V$, then we denote the neighbourhood of $x$, that is the set of the vertices adjacent to $x$, by $N(x)$. For a set $S \subset V$, we denote the neighbourhood of $S$, that is the set of vertices in $V \backslash S$ that are adjacent to at least an element of $S$, by $N(S)$. We denote $|V|$ by $|G|$ and $|E|$ by $||G||$.

We call a vertex of degree $d$, at least $d$, and at most $d$, a \emph{$d$-vertex}, a \emph{$d^+$-vertex}, and a \emph{$d^-$-vertex} respectively. We call a cycle of length $l$ an \emph{$l$-cycle}, and by extension a face of length $l$ an \emph{$l$-face}.

Let ${\cal P}_4$ be the class of triangle-free planar graphs. Let $G = (V,E)$ be a counter-example to Theorem~\ref{main} with the minimum order. Let $n = |G|$ and $m = ||G||$. We will use the schemes presented in Observations~\ref{abg}--\ref{o2} many times throughout this paper.

\begin{obs} \label{abg}
Let $\alpha$, $\beta$, $\gamma$ be integers satisfying $\alpha \ge 1$, $\beta \ge 0$, $\gamma \ge 0$.
Let $H^* \in {\cal P}_4$ be a graph with $|H^*| = n - \alpha$ and $||H^*|| \le m - \beta$.
By minimality of $G$, $H^*$ admits an induced linear forest of order at least $\frac{9}{11}(n-\alpha) - \frac{2}{11}(m - \beta)$.
Given an induced linear forest $F^*$ of $H^*$ of order $|F^*| \ge \frac{9}{11}(n-\alpha) - \frac{2}{11}(m - \beta)$, if there is an induced linear forest $F$ of $G$ of order $|F| \ge |F^*| + \gamma$, then as $|F| < \frac{9}{11}n - \frac{2}{11}m$, we have $\gamma < \frac{9}{11}\alpha - \frac{2}{11}\beta$.
\end{obs}

\begin{figure}
\begin{center}
\begin{tikzpicture}[scale=1.5]

\draw [fill=gray!20, draw = gray!20] (-3,0) circle (1) ;

\draw [very thick] [-,>=latex] (-3.3,0) to[out = 45, in = -90] (-3,0.7);
\draw [very thick] [-,>=latex] (-3.5,-0.4) to[out = 0, in = 180] (-2.9,-0.7) to[out = 0, in = 0] (-2.9,-0.1);
\draw [very thick] [-,>=latex] (-2.5,0.2)  to[out = 90, in = 0] (-2.7, 0.4) to[out = 180, in = 0] (-2.7,0.7);
\draw (-3.05,0) node {$L$};
\draw [gray] (-2,-1) node {$M$};

\end{tikzpicture}~~~~~~~
\begin{tikzpicture}[scale=1.5]

\draw [fill=gray!20, draw = gray!20] (-3,0) circle (1) ;

\draw [very thick] [-,>=latex] (-3.3,0) to[out = 45, in = -90] (-3,0.9);
\draw [-,>=latex] (-3,0.9) to (-3,1.1);
\draw [dashed] (-3,1.1) to (-3,1.5);
\draw [very thick] [-,>=latex] (-3.5,-0.4) to[out = 0, in = 180] (-2.9,-0.7) to[out = 0, in = 0] (-2.9,-0.1);
\draw [very thick] [-,>=latex] (-2.5,0.2)  to[out = 90, in = 0] (-2.7, 0.4) to[out = 180, in = 0] (-2.7,0.7);
\draw (-3.05,0) node {$L$};
\draw [gray] (-2,-1) node {$M$};
\draw [fill=black] (-3,0.9) circle (1.5pt) ;
\draw [fill=black] (-3,1.1) circle (1.5pt) ;
\draw (-3,0.85) node [right] {$u$};
\draw (-3,1.1) node [right] {$v$};

\end{tikzpicture}
\caption{The situation in Observation \ref{o1} (left) and Observation \ref{o2} (right). Thick lines are paths, normal lines are edges and dashed lines are edges that may or may not be present.}\label{figobs}
\end{center}
\vspace{-0.5cm}
\end{figure}
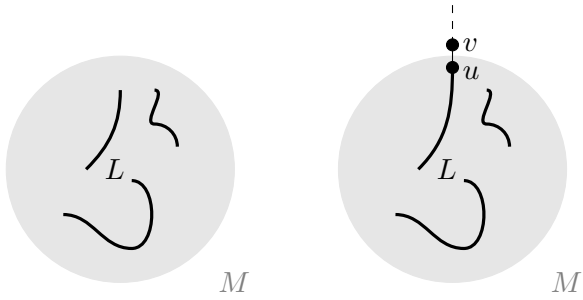

\begin{obs} \label{o1}
Suppose $L \subset V$ induces a linear forest in $G$, and $M$ is a set of vertices such that $M \cap L = \emptyset$ and $M \supset N(L)$.
Let $G' = G - M - L$. See Figure~\ref{figobs} (left) for an illustration.
By minimality of $G$, $G'$ admits a linear forest $F'$ with $|F'| \ge \frac{9}{11}|G'| - \frac{2}{11}||G'||$. Observe that $F = G[V(F') \cup L]$ is an induced linear forest of $G$.
As $G$ is a counter-example to Theorem~\ref{main}, $|F| < \frac{9}{11}|G| - \frac{2}{11}||G||$. 
Therefore $|L| = |F| - |F'| < \frac{9}{11}(|M|+|L|) - \frac{2}{11}(||G|| - ||G'||)$.
\end{obs}

\begin{obs} \label{o2}
Suppose $L \subset V$ induces a linear forest in $G$. Suppose there is a set of vertices $M$ and two vertices $u \in L$ and $v$ such that $M \cap L = \emptyset$, $\{v\} = N(L) \setminus M$, and $\{u\} = N(v) \cap L$. 
Let $G' = G - M - L$. Suppose $v$ is a $1^-$-vertex in $G'$ and $u$ is a $1^-$-vertex in $G[L]$. See Figure~\ref{figobs} (right) for an illustration.
By minimality of $G$, $G'$ admits a linear forest $F'$ with $|F'| \ge \frac{9}{11}|G'| - \frac{2}{11}||G'||$. Observe that $F = G[V(F') \cup L]$ is an induced linear forest of $G$.
As $G$ is a counter-example to Theorem~\ref{main}, $|F| < \frac{9}{11}|G| - \frac{2}{11}||G||$. 
Therefore $|L| = |F| - |F'| < \frac{9}{11}(|M|+|L|) - \frac{2}{11}(||G|| - ||G'||)$.
\end{obs}

Now we want to prove some structural properties of $G$, so that we can later show that the counter-example $G$ does not exist, and thus that Theorem~\ref{main} is true.
First note that $G$ is connected, otherwise one of its components would be a smaller counter-example to Theorem~\ref{main}. Then note that every vertex of $G$ has degree at most $4$. Otherwise, by considering a vertex of degree at least $5$ and by Observation~\ref{abg} applied to $H^* = G - v$ with $(\alpha,\beta,\gamma) = (1,5,0)$ and $F = F^*$, we have $0 < \frac{9}{11} - 5\frac{2}{11}$, a contradiction. 

\begin{figure}
\begin{center}
\begin{tikzpicture}[scale=1.5]

\draw [fill=gray!20, draw = gray!20] (-3,0) circle (1) ;

\draw [very thick] [-,>=latex] (-3.3,0) to[out = 45, in = -90] (-3,0.9);
\draw [dashed] (-3,1.1) to (-3.2,1.3);
\draw [dashed] (-3,1.1) to (-2.8,1.3);
\draw [-,>=latex] (-3,0.9) to (-3,1.1);
\draw [very thick] [-,>=latex] (-3.5,-0.4) to[out = 0, in = 180] (-2.9,-0.7) to[out = 0, in = 0] (-2.9,-0.1);
\draw [very thick] [-,>=latex] (-2.5,0.2)  to[out = 90, in = 0] (-2.7, 0.4) to[out = 180, in = 0] (-2.7,0.7);
\draw (-3,0) node {$P$};
\draw [gray] (-2,-1) node {$N$};
\draw [fill=black] (-3,0.9) circle (1.5pt) ;
\draw [fill=black] (-3,1.1) circle (1.5pt) ;
\draw (-3,0.85) node [right] {$u$};
\draw (-3,1.1) node [right] {$v$};
\end{tikzpicture}~~~~~
\begin{tikzpicture}[scale=1.5]

\draw [fill=gray!20, draw = gray!20] (3,0) circle (1) ;

\draw [very thick] [-,>=latex] (2.1,0) to[out = 0, in = -90] (3,0.9);
\draw [dashed] (3,1.1) to (3.2,1.3);
\draw [dashed] (3,1.1) to (2.8,1.3);
\draw [dashed] (1.8,0) to (1.6,0.2);
\draw [dashed] (1.8,0) to (1.6,-0.2);
\draw [-,>=latex] (3,0.9) to (3,1.1);
\draw [-,>=latex] (2.1,0) to (1.9,0);
\draw [very thick] [-,>=latex] (2.5,-0.4) to[out = 0, in = 180] (3.1,-0.7) to[out = 0, in = 0] (3.1,-0.1);
\draw [very thick] [-,>=latex] (3.5,0.2)  to[out = 90, in = 0] (3.3, 0.4) to[out = 180, in = 0] (3.3,0.7);
\draw (3,0) node {$P$};
\draw [gray] (4,-1) node {$N$};
\draw [fill=black] (3,0.9) circle (1.5pt) ;
\draw [fill=black] (3,1.1) circle (1.5pt) ;
\draw (3,0.85) node [right] {$u_0$};
\draw (3,1.1) node [right] {$v_0$};
\draw [fill=black] (2.1,0) circle (1.5pt) ;
\draw [fill=black] (1.9,0) circle (1.5pt) ;
\draw (2.3,0) node [below] {$u_1$};
\draw (1.9,0) node [below] {$v_1$};

\end{tikzpicture}
\caption{A simple chain (left) and a double chain (right).}\label{figschain}
\end{center}
\vspace{-0.5cm}
\end{figure}
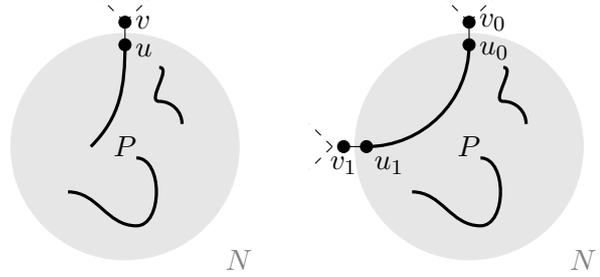

Let us define the notion of a \emph{chain} of $G$ (or \emph{simple chain}) of $G$ which is a quadruplet $C = (P,N,u,v)$ such that:

\begin{itemize}
\item $P \subset V$, $N \subset V \setminus P$, $u \in P$, and $v \in V \backslash (N \cup P)$;

\item $G[P]$ is a linear forest;

\item vertex $u$ is a $1^-$-vertex of $G[P]$, and $N(v) \cap P = \{u\}$;

\item $N(P) \subset N \cup \{v\}$ in $G$;

\item vertex $v$ is a $2^-$-vertex in $G - (N \cup P)$.
\end{itemize}

See Figure~\ref{figschain} (left) for an illustration. We will use the following notation for a chain $C = (P,N,u,v)$ of $G$: 

\begin{itemize}
\item $|C| = |P| + |N|$;

\item $G - C = G - (N \cup P)$;

\item $d(C)$ is the degree of $v$ in $G - C$ (thus $d(C) \le 2$);

\item $||C|| = ||G|| - ||G - C||$.
\end{itemize}

We will now prove the following lemma:

\begin{lemm} \label{schain}
For every chain $C = (P,N,u,v)$ of $G$, $|P| < \frac{9}{11}|C| - \frac{2}{11}(||C||-\frac{1}{2})$.
\end{lemm}

\begin{proof}
Let us consider by contradiction a chain $C = (P,N,u,v)$ such that $|P| \ge \frac{9}{11}|C| - \frac{2}{11}(||C||-\frac{1}{2})$ maximizing $|C|$.

\begin{itemize}
\item Suppose $d(C) = 0$. Let $G' = (G - C) - v$. The set $P \cup \{v\}$ induces a linear forest, and its neighbourhood is a subset of $N$. By Observation~\ref{o1} applied to $L = P \cup \{v\}$ and $M = N$, we have $|P| + 1 < \frac{9}{11}(|C| + 1) - \frac{2}{11}||C||$. As $|P| \ge \frac{9}{11}|C|  - \frac{2}{11}(||C||-\frac{1}{2})$, it follows that $1 < \frac{9}{11} -  \frac{2}{11}\frac{1}{2}$, a contradiction.

\item Suppose $d(C) = 1$. The set $P$ induces a linear forest, and its neighbourhood is a subset of $N \cup \{v\}$. Furthermore, $N(v) \cap P = \{u\}$, $N(u) \cap (G - C) = \{v\}$, and $u$ and $v$ are $1$-vertices in $P$ and $G - C$ respectively. By Observation~\ref{o2} applied to $L = P$ and $M = N $, we have $|P| < \frac{9}{11}|C| - \frac{2}{11}||C||$, thus $|P|  < \frac{9}{11}|C|  - \frac{2}{11}(||C||-\frac{1}{2})$, a contradiction.

%

\item Suppose $d(C) = 2$. Let $w_0$ and $w_1$ be the neighbours of $v$ in $G - C$. 

\begin{itemize}
\item Suppose one of the $w_i$'s, say $w_0$, has degree $1$ in $G - C$. Let $G' = (G - C) - \{v, w_0, w_1\}$. The set $P \cup \{v,w_0\}$ induces a linear forest and its neighbourhood is a subset of $N \cup \{w_1\}$. By Observation~\ref{o1} applied to $L = P \cup \{v,w_0\}$ and $M = N \cup \{w_1\}$, we have $|P| + 2 < \frac{9}{11}(|C| + 3) - \frac{2}{11}(||C||+2)$. As $|P|  \ge \frac{9}{11}|C|  - \frac{2}{11}(||C||-\frac{1}{2})$, it follows that $2 < \frac{9}{11}3 -  \frac{2}{11}\frac{5}{2}$, a contradiction.

\item Suppose the $w_i$'s both have degree $2$ in $G - C$. Observe that they are not adjacent since $G$ is triangle-free. The set $P \cup \{v\}$ induces a linear forest, and its neighbourhood is a subset of $N \cup \{w_0,w_1\}$. Furthermore, $N(w_1) \cap (P \cup \{v\}) = \{v\}$, $N(v) \cap ((G - C )- \{v,w_0\}) = \{w_1\}$, and $v$ and $w_1$ are $1$-vertices in $P \cup \{v\}$ and $G - C - \{v,w_0\}$ respectively. By Observation~\ref{o2} applied to $L = P \cup \{v\}$ and $M = N \cup \{w_0\}$, we have $|P| + 1 < \frac{9}{11}(|C| + 2) - \frac{2}{11}(||C||+3)$. As $|P|  \ge \frac{9}{11}|C|  - \frac{2}{11}(||C||-\frac{1}{2})$, it follows that $1 < \frac{9}{11}2 -  \frac{2}{11}\frac{7}{2}$, a contradiction.

\item Suppose the $w_i$'s both have degree $4$ in $G - C$. Again they are not adjacent since $G$ is triangle-free.  
The set $P \cup \{v\}$ induces a linear forest, and its neighbourhood is a subset of $N \cup \{w_1,w_0\}$. By Observation~\ref{o1} applied to $L = P \cup \{v\}$ and $M = N \cup \{w_0,w_1\}$, we have $|P| + 1 < \frac{9}{11}(|C| + 3) - \frac{2}{11}(||C||+8)$. As $|P| \ge \frac{9}{11}|C|  - \frac{2}{11}(||C||-\frac{1}{2})$, it follows that $1 < \frac{9}{11}3 -  \frac{2}{11}\frac{17}{2}$, a contradiction.

\item Suppose one of the $w_i$'s, say $w_0$, is a $3^-$-vertex in $G - C$ and the other one is a $3^+$-vertex in $G - C$. Let $C' = (P \cup\{v\}, N \cup \{w_1\}, v, w_0)$. Then $C'$ is a chain of $G$, and by maximality of $|C|$, we have $|P| + 1 < \frac{9}{11}(|C| + 2) - \frac{2}{11}(||C|| + \frac{7}{2})$. As $|P| \ge \frac{9}{11}|C| - \frac{2}{11}(||C||- \frac{1}{2})$, it follows that $1 < \frac{9}{11}2 -  \frac{2}{11}4$, a contradiction.
\end{itemize}
\end{itemize}
\vspace{-0.8cm}
\end{proof}

Let us now define a new notion quite similar to the notion of chain. A \emph{double chain} of $G$ is a sextuplet $C = (P,N,u_0,u_1,v_0,v_1)$, so that:

\begin{itemize}
\item $P \subset V$, $N \subset V \setminus P$, $u_0 \in P$, $u_1 \in P$, $v_0 \in V \backslash (N \cup P)$ and $v_1 \in V \backslash (N \cup P)$;

\item $v_0 \ne v_1$;

\item $G[P]$ is a linear forest;

\item $u_0$ and $u_1$ are $1^-$-vertices of $G[P]$ if they are distinct, a $0$-vertex of $G[P]$ if they are equal, and for $i \in \{0,1\}$, $N(v_i) \cap P = \{u_i\}$;

\item $N(P) \subset N \cup \{v_0\} \cup \{v_1\}$;

\item $v_0$ and $v_1$ are $2^-$-vertices in $G - (N \cup P)$.
\end{itemize}

See Figure~\ref{figschain} (right) for an illustration. We will use the following notations for a double chain $C = (P,N,u_0,u_1,v_0,v_1)$ of $G$: 

\begin{itemize}
\item $|C| = |P| + |N|$;

\item $G - C = G - (N \cup P)$;

\item $d_0(C)$ is the degree of $v_0$ in $G - C$ (thus $d_0(C) \le 2$);

\item $d_1(C)$ is the degree of $v_1$ in $G - C$ (thus $d_1(C) \le 2$);

\item $||C|| = ||G|| - ||G - C||$.
\end{itemize}

A double chain $C = (P,N,u_0,u_1,v_0,v_1)$ of $G$ such that $v_0$ and $v_1$ are on different components of $G - C$ is called a \emph{separating double chain} of $G$.

We will now prove the following lemmas:

\begin{lemm} \label{dchain}
For every double chain $C = (P,N,u_0,u_1,v_0,v_1)$ of $G$, $|P| < \frac{9}{11}|C| - \frac{2}{11}(||C||-3)$.
\end{lemm}

\begin{proof}
Let us consider by contradiction a double chain $C = (P,N,u_0,u_1,v_0,v_1)$ such that $|P| \ge \frac{9}{11}|C| - \frac{2}{11}(||C||-3)$ maximizing $|C|$.

Suppose first that $v_0$ and $v_1$ are not adjacent.

\begin{itemize}
\item Suppose $d_0(C) = 0$. Then $(P \cup \{v_0\},N,u_1,v_1)$ is a simple chain of $G$. By Lemma~\ref{schain}, $|P|+1 < \frac{9}{11}(|C|+1) - \frac{2}{11}(||C||-\frac{1}{2})$. As $|P| \ge \frac{9}{11}|C| - \frac{2}{11}(||C||-3)$, we have $1 < \frac{9}{11}-\frac{5}{2}\frac{2}{11}$, a contradiction.

\item Suppose $d_0(C) = 1$. Let $w$ be the neighbour of $v_0$ in $G - C$. Then $(P \cup \{v_0\},N \cup \{w\},u_1,v_1)$ is a simple chain of $G$. By Lemma~\ref{schain}, $|P|+1 < \frac{9}{11}(|C|+2) - \frac{2}{11}(||C|| +\frac{1}{2})$. As $|P| \ge \frac{9}{11}|C| - \frac{2}{11}(||C||-3)$, we have  $1 < \frac{9}{11}2-\frac{2}{11}\frac{7}{2}$, a contradiction.

\item Suppose $d_0(C) = 2$. Let $w_0$ and $w_1$ be the neighbours of $v_0$ in $G - C$.

\begin{itemize}
\item Suppose one of the $w_i$'s, say $w_0$, has degree $1$ in $G - C$.
 Then $(P \cup \{v_0,w_0\},N \cup \{w_1\},u_1,v_1)$ is a simple chain of $G$. By Lemma~\ref{schain}, $|P|+2 < \frac{9}{11}(|C|+3) - \frac{2}{11}(||C|| + \frac{3}{2})$. As $|P| \ge \frac{9}{11}|C| - \frac{2}{11}(||C||-3)$, we have $2 < \frac{9}{11}3-\frac{2}{11}\frac{9}{2}$, a contradiction.
 
\item
Suppose that the $w_i$'s both have degree $2$ in $G - C$. Note that they are not adjacent since $G$ is triangle-free. They may, however, be adjacent to $v_1$. 

Suppose both of the $w_i$'s are adjacent to $v_1$. The set $P \cup \{v_0,v_1\}$ induces a linear forest in $G$, and its neighbourhood is a subset of $N \cup \{w_0,w_1\}$. By Observation~\ref{o1} applied to $L = P \cup \{v_0,v_1\}$ and $M = N \cup \{w_0,w_1\}$, we have $|P| + 2 < \frac{9}{11}(|C| + 4) - \frac{2}{11}(||C||+4)$. As $|P| \ge \frac{9}{11}|C|  - \frac{2}{11}(||C||-3)$, it follows that $2 < \frac{9}{11}4 -  \frac{2}{11}7$, a contradiction.

Therefore one of the $w_i$'s, say $w_0$, is not adjacent to $v_1$. Let $x$ be the neighbour of $w_0$ in $G - C$ distinct from $v_0$. Suppose $x$ has degree $4$ in $G - C - \{v_0,w_1\}$. Now $(P \cup \{v_0,w_0\},N \cup \{w_1,x\},u_1,v_1)$ is a chain of $G$. By Lemma~\ref{schain}, $|P|+2 < \frac{9}{11}(|C|+4) - \frac{2}{11}(||C|| + \frac{13}{2})$. As $|P| \ge \frac{9}{11}|C| - \frac{2}{11}(||C||-3)$, we have $2 < \frac{9}{11}4-\frac{2}{11}\frac{19}{2}$, a contradiction. Therefore $x$ is a $3^-$-vertex in $G - C - \{v_0,w_1\}$. Then $(P \cup \{v_0,w_0\},N \cup \{w_1\},w_0,u_1,x,v_1)$ is a double chain of $G$, so by maximality of $|C|$, $|P| + 2 < \frac{9}{11}(|C|+3) - \frac{2}{11}(||C|| + 4 -3)$. As $|P| \ge \frac{9}{11}|C| - \frac{2}{11}(||C||- 3)$, we have $2 < \frac{9}{11}3 -  \frac{2}{11}4$, a contradiction.

\item Suppose that the $w_i$'s are $4$-vertices in $G-C$. Now $(P \cup \{v_0\},N \cup \{w_0,w_1\},u_1,v_1)$ is a chain of $G$. By Lemma~\ref{schain}, $|P|+1 < \frac{9}{11}(|C|+3) - \frac{2}{11}(||C|| + \frac{15}{2})$. As $|P| \ge \frac{9}{11}|C| - \frac{2}{11}(||C||-3)$, we have $1 < \frac{9}{11}3-\frac{2}{11}\frac{21}{2}$, a contradiction.

\item Suppose one of the $w_i$, say $w_0$, is a $3^-$-vertex in $G - C$ and the other one is a $3^+$-vertex in $G - C$. Then $(P \cup \{v_0\},N \cup \{w_1\},v_0,u_1,w_0,v_1)$ is a double chain of $G$. By maximality of $|C|$, $|P| + 1 < \frac{9}{11}(|C|+2) - \frac{2}{11}(||C|| + 4 -3)$. As $|P| \ge \frac{9}{11}|C| - \frac{2}{11}(||C|| -3)$, we have $1 < \frac{9}{11}2 -  \frac{2}{11}4$, a contradiction. 
\end{itemize}
\end{itemize}

Now $v_0$ and $v_1$ are adjacent.

\begin{itemize}
\item Suppose $d_0(C) = 1$ or $d_1(C) = 1$, say $d_0(C) = 1$. The set $P \cup \{v_0\}$ induces a linear forest, and its neighbourhood is a subset of $N \cup \{v_1\}$. By Observation~\ref{o1} applied to $L = P \cup \{v_0\}$ and $M = N \cup \{v_1\}$, we have $|P| + 1 < \frac{9}{11}(|C| + 2) - \frac{2}{11}(||C||+1)$. As $|P| \ge \frac{9}{11}|C|  - \frac{2}{11}(||C||-3)$, it follows that $1 < \frac{9}{11}2 -  \frac{2}{11}4$, a contradiction.

\item Now $d_0(C) = 2$ and $d_1(C) = 2$. Let $w$ be the neighbour of $v_0$ in $(G - C )- v_1$. Note that $w$ is not adjacent to $v_1$, otherwise $v_0v_1w$ would be a triangle in $G$.

Suppose $w$ is a $2^-$-vertex in $(G - C )$.
The set $P \cup \{v_0\}$ induces a linear forest, and its neighbourhood is a subset of $N \cup \{v_1,w\}$. Furthermore, $N(w) \cap (P \cup \{v_0\}) = \{v_0\}$, $N(v_0) \cap V((G - C) - \{v_1\}) = \{w\}$, and $v_0$ and $w$ are $1^-$-vertices in $G[P \cup \{v_0\}]$ and $(G - C) - \{v_0\}$ respectively. By Observation~\ref{o2} applied to $L = P \cup \{v_0\}$ and $M = N \cup \{v_1\}$, we have $|P| + 1 < \frac{9}{11}(|C| + 2) - \frac{2}{11}(||C||+3)$. As $|P|  \ge \frac{9}{11}|C|  - \frac{2}{11}(||C||-3)$, it follows that $1 < \frac{9}{11}2 - \frac{2}{11}6$, a contradiction.

Now $w$ is a $3^+$-vertex in $(G - C)$. The set $P \cup \{v_0\}$ induces a linear forest, and its neighbourhood is a subset of $N \cup \{v_1,w\}$. By Observation~\ref{o1} applied to $L = P \cup \{v_0\}$ and $M = N \cup \{w\}$, we have $|P| + 1 < \frac{9}{11}(|C| + 3) - \frac{2}{11}(||C||+5)$. As $|P| \ge \frac{9}{11}|C|  - \frac{2}{11}(||C||-3)$, it follows that $1 < \frac{9}{11}3 -  \frac{2}{11}8$, a contradiction.
\end{itemize}
\vspace{-0.8 cm}
\end{proof}

\begin{lemm} \label{sepdchain}
For every separating double chain $C = (P,N,u_0,u_1,v_0,v_1)$ of $G$, $|P| < \frac{9}{11}|C| - \frac{2}{11}(||C||-1)$.
\end{lemm}

\begin{proof}
Let us consider by contradiction a separating double chain $C = (P,N,u_0,u_1,v_0,v_1)$ such that $|P| \ge \frac{9}{11}|C| - \frac{2}{11}(||C||-1)$ maximizing $|C|$.

\begin{itemize}
\item Suppose $d_0(C) = 0$.  Then $(P \cup \{v_0\},N,u_1,v_1)$ is a simple chain of $G$. By Lemma~\ref{schain}, $|P|+1 < \frac{9}{11}(|C|+1) - \frac{2}{11}(||C||-\frac{1}{2})$. As $|P| \ge \frac{9}{11}|C| - \frac{2}{11}(||C||-1)$, we have $1 < \frac{9}{11} - \frac{2}{11}\frac{1}{2}$, a contradiction.

\item Suppose $d_0(C) = 1$. Let $w$ be the neighbour of $v_0$ in $G - C$. Suppose $w$ is a $4^+$-vertex in $G - C$. Then $(P \cup \{v_0\},N \cup \{w\},u_1,v_1)$ is a simple chain of $G$. By Lemma~\ref{schain}, $|P|+1 < \frac{9}{11}(|C|+2) - \frac{2}{11}(||C|| +\frac{7}{2})$. As $|P| \ge \frac{9}{11}|C|  - \frac{2}{11}(||C||-1)$, we have $1 < \frac{9}{11}2 -  \frac{2}{11}\frac{9}{2}$, a contradiction.

Now $w$ is a $3^-$-vertex in $G - C$. Let $C' = (P \cup\{v_0\}, N, v_0, u_1, w, v_1)$. One can see that $C'$ is a separating double chain of $G$, and by maximality of $|C|$, $|P| + 1 < \frac{9}{11}(|C| + 1) - \frac{2}{11}(||C||)$. As $|P| \ge \frac{9}{11}|C| - \frac{2}{11}(||C||- 1)$, we have $1 < \frac{9}{11} -  \frac{2}{11}$, a contradiction.

\item Suppose $d_0(C) = 2$. Let $w_0$ and $w_1$ be the neighbours of $v_0$ in $G - C$. 

\begin{itemize}
\item Suppose one of the $w_i$'s, say $w_0$, has degree $1$ in $G - C$.  We have a simple chain $(P \cup\{v_0,w_0\}, N \cup \{w_1\}, u_1, v_1)$. By Lemma~\ref{schain}, $|P|+2 < \frac{9}{11}(|C|+3) - \frac{2}{11}(||C|| + \frac{3}{2})$. As $|P| \ge \frac{9}{11}|C| - \frac{2}{11}(||C||- 1)$, we have $2 < \frac{9}{11}3 -\frac{2}{11} \frac{5}{2}$, a contradiction.

\item Suppose the $w_i$'s both have degree $2$ in $G - C$. Note that they are not adjacent since $G$ is triangle-free. Let $x$ be the second neighbour of $w_0$ in $G - C$. Suppose $x$ is a $4$-vertex in $G - C - \{w_1\}$. Then $(P \cup\{v_0,w_0\}, N \cup \{w_1,x\}, u_1, v_1)$ is a simple chain of $G$. By Lemma~\ref{schain}, $|P|+2 < \frac{9}{11}(|C|+4) - \frac{2}{11}(||C|| +\frac{13}{2})$. As $|P| \ge \frac{9}{11}|C| - \frac{2}{11}(||C||- 1)$, we have $2 < \frac{9}{11}4 - \frac{2}{11}\frac{15}{2}$, a contradiction. Now $x$ is a $3^-$-vertex in $G - C - \{w_1\}$, so $(P \cup\{v_0,w_0\}, N \cup \{w_1\}, w_0, u_1, x, v_1)$ is a separating double chain of $G$. By maximality of $|C|$, $|P| + 2 < \frac{9}{11}(|C| + 3) - \frac{2}{11}(||C|| + 3)$. As $|P| \ge \frac{9}{11}|C| - \frac{2}{11}(||C||- 1)$, we have $2 < \frac{9}{11}3 -  \frac{2}{11}4$, a contradiction.

\item Suppose the $w_i$'s have degree $4$ in $G - C$. Again they are not adjacent since $G$ is triangle-free. We have a simple chain $(P \cup\{v_0\}, N \cup \{w_0,w_1\}, u_1, v_1)$. By Lemma~\ref{schain}, $|P|+1 < \frac{9}{11}(|C|+3) - \frac{2}{11}(||C|| + \frac{15}{2})$. As $|P| \ge \frac{9}{11}|C| - \frac{2}{11}(||C||- 1)$, we have $1 < \frac{9}{11}3 - \frac{2}{11}\frac{17}{2}$, a contradiction. 

\item Suppose one of the $w_i$'s, say $w_0$, is a $3^-$-vertex in $G - C$ and the other one is a $3^+$-vertex in $G - C$. Then $(P \cup\{v_0\}, N \cup \{w_1\}, v_0, u_1, w_0, v_1)$ is a separating double chain. By maximality of $|C|$, $|P| + 1 < \frac{9}{11}(|C| + 2) - \frac{2}{11}(||C|| + 3)$. As $|P| \ge \frac{9}{11}|C| - \frac{2}{11}(||C||- 1)$, we have $1 < \frac{9}{11}2 -  \frac{2}{11}4$, a contradiction.
\end{itemize}
\end{itemize}
\vspace{-0.8 cm}
\end{proof}

Let us now prove some lemmas on the structure of $G$.

\begin{lemm} \label{degge3}
Graph $G$ has no $2^-$-vertex.
\end{lemm}

\begin{proof}
As $G$ is connected, if it has a $0$-vertex, then $G$ is the graph with one vertex and it satisfies Theorem~\ref{main}, a contradiction.

By contradiction, suppose $u \in V$ is a $1$-vertex. Let $v$ be the neighbour of $u$. If $v$ is a $3^-$-vertex in $G$, then $(\{u\},\emptyset,u,v)$ is a chain of $G$, thus by Lemma~\ref{schain}, $1 < \frac{9}{11} - \frac{2}{11}(1 - \frac{1}{2})$, a contradiction.

Now $v$ is a $4$-vertex. Let $H^* = G - \{u,v\}$. Graph $H^*$ has $n-2$ vertices and $m-4$ edges. Adding vertex $u$ to any induced linear forest of $H^*$ leads to an induced linear forest of $G$. By Observation~\ref{abg} applied to $(\alpha, \beta, \gamma) = (2,4,1)$, $1 < \frac{9}{11}2 - \frac{2}{11}4$, a contradiction.

Therefore $G$ has no $1^-$-vertex.
Suppose now that $u \in V$ is a $2$-vertex.
Let $v_0$ and $v_1$ be the two neighbours of $u$. 

Suppose $v_0$ or $v_1$, say $v_1$, is a $3^-$-vertex. We have a simple chain $(\{u\}, \{v_0\}, u, v_1)$. By Lemma~\ref{schain}, $1 < \frac{9}{11}2 - \frac{2}{11}(4-\frac{1}{2})$, a contradiction.

Now $v_0$ and $v_1$ are $4$-vertices. Let $H^* = G - \{u,v_0,v_1\}$. Graph $H^*$ has $n-3$ vertices and $m-8$ edges. Adding vertex $u$ to any induced linear forest of $H^*$ leads to an induced linear forest of $G$. By Observation~\ref{abg} applied to $(\alpha, \beta, \gamma) = (3,8,1)$, $1 < \frac{9}{11}3 - \frac{2}{11}8$, a contradiction.
\end{proof}

\begin{lemm} \label{3-344}
Graph $G$ has no $3$-vertex adjacent to another $3$-vertex and two $4$-vertices.
\end{lemm}

\begin{proof}
By contradiction, suppose $G$ has a $3$-vertex $u$, adjacent to a $3$-vertex $v$ and two $4$-vertices $w_0$ and $w_1$. We have a simple chain $(\{u\},\{w_0,w_1\},u,v)$. By Lemma~\ref{schain}, $1 < \frac{9}{11}3 - \frac{2}{11}(9 - \frac{1}{2})$, a contradiction.
\end{proof}

\begin{lemm}\label{3-334}
Graph $G$ has no $3$-vertex adjacent to two other $3$-vertices and a $4$-vertex.
\end{lemm}

\begin{proof}
Let $u$ be a $3$-vertex adjacent to two $3$-vertices $v_0$ and $v_1$, and to a $4$-vertex $w$. Let $x_0$ and $x_1$ be the two neighbours of $v_0$ distinct from $u$. Note that $x_0$ and $x_1$ are $3^+$-vertices in $G$ by Lemma~\ref{degge3}, and thus $1^+$-vertices in $G' = G - \{u,w,v_0\}$ since they are not adjacent to $u$. Note that $x_0$ and $x_1$ may be adjacent to $w$.

Suppose that $x_0$ and $x_1$ are $2^+$-vertices in $G'$, or that one is a $3$-vertex and the other a $1^+$-vertex. We have a simple chain $(\{u,v_0\}, \{x_0,x_1,w\}, u, v_1)$ in $G$. By Lemma~\ref{schain}, $2 < \frac{9}{11}5 - \frac{2}{11}(12 - \frac{1}{2})$, a contradiction.

Suppose one of the $x_i$'s, say $x_0$, is a $2^+$-vertex in $G'$, and the other one is a $1$-vertex in $G'$. We have a double chain $(\{u,v_0\},\{w,x_0\},u,v_0,v_1,x_1)$. By Lemma~\ref{dchain}, $2 < \frac{9}{11}4 - \frac{2}{11}(10 - 3)$, a contradiction.

Now the $x_i$'s are $1$-vertices in $G'$. By Lemma~\ref{degge3}, the $x_i$'s are $3$-vertices in $G$, and thus are both adjacent to $w$. By planarity of $G$, one of the $x_i$'s, say $x_0$, is not adjacent to $v_1$. Let $y$ be the neighbour of $x_0$ in $G'$. By Lemmas~\ref{degge3} and~\ref{3-344}, $y$ is a $3$-vertex in $G$. We have a simple chain $(\{u,v_0,x_0\},\{w,v_1,x_1\},x_0,y)$. By Lemma~\ref{schain}, $3 < \frac{9}{11}6 - \frac{2}{11}(11 - \frac{1}{2}$, a contradiction.
\end{proof}

\begin{lemm}\label{3-3}
Graph $G$ has no two adjacent $3$-vertices.
\end{lemm}

\begin{proof}
By Lemma~\ref{degge3}, every vertex in $G$ has degree $3$ or $4$. By Lemmas~\ref{3-344} and \ref{3-334}, there is no $3$-vertex adjacent to a $3$-vertex and a $4$-vertex in $G$.
Suppose by contradiction that there are two adjacent $3$-vertices in $G$. Then as $G$ is connected, $G$ only has $3$-vertices.

Suppose there is a $4$-cycle $u_0u_1u_2u_3$ in $G$. For all $i$, let $v_i$ be the third neighbour of $u_i$. Since $G$ has no triangle, the only vertices among the $u_i$'s and $v_i$'s that may not be distinct are $v_0$ and $v_2$ on the one hand, and $v_1$ and $v_3$ on the other hand. Suppose $v_0 = v_2$ and $v_1 = v_3$. Let $H^* = G - \{u_0,u_1,u_2,u_3\}$. Graph $H^*$ has $n-4$ vertices and $m-8$ edges. As $v_0$ and $v_1$ are $1$-vertices in $H^*$, separated by $u_0u_1u_2u_3$ in $G$, adding vertices $u_0$ and $u_1$ to any induced linear forest of $H^*$ leads to an induced linear forest of $G$. By Observation~\ref{abg} applied to $(\alpha, \beta, \gamma) = (4,8,2)$, we have $2 < \frac{9}{11}4 - \frac{2}{11}8$, a contradiction. Now w.l.o.g. $v_0$ and $v_2$ are distinct. We have a double chain $(\{u_0,u_1,u_2\}, \{u_3,v_1\}, u_0, u_2, v_0, v_2)$. By Lemma~\ref{dchain}, $3 < \frac{9}{11}5 - \frac{2}{11}(9 - 3)$, a contradiction.

Now there is no $4$-cycle in $G$. Suppose there is a $5$-cycle $u_0u_1u_2u_3u_4$ in $G$. For all $i$, let $v_i$ be the third neighbour of $u_i$. Now all the $v_i$'s are distinct, otherwise there is a $4$-cycle and we fall into the previous case. We have a double chain $(\{u_0,u_1,u_2,u_3\}, \{u_4,v_1,v_2\}, u_0, u_3, v_0, v_3)$. By Lemma~\ref{dchain}, $4 < \frac{9}{11}7 - \frac{2}{11}(14 - 3)$, a contradiction.

Now $G$ is a $3$-regular planar graph with girth at least $6$, which contradicts Euler's formula.
\end{proof}

\begin{lemm} \label{4433}
There is no $4$-cycle with at least two $3$-vertices in $G$.
\end{lemm}

\begin{proof}
By contradiction, suppose there is such a $4$-cycle $u_0u_1u_2u_3$. By Lemmas~\ref{degge3} and~\ref{3-3}, this cycle has exactly two $3$-vertices and two $4$-vertices, and the two $3$-vertices are not adjacent. W.l.o.g. $u_0$ and $u_2$ are $3$-vertices, and $u_1$ and $u_3$ are $4$-vertices.
Let $v_0$ and $v_2$ be the third neighbours of $u_0$ and $u_2$ respectively. By Lemma~\ref{3-3}, $v_0$ and $v_2$ are $4$-vertices.

Suppose that $u_0$ and $u_2$ have three neighbours in common, $u_1$, $u_3$, and $v = v_0 = v_2$. Let $H^* = G - \{u_0,u_1,u_2,u_3,v\}$. Graph $H^*$ has $n-5$ vertices and $m-12$ edges. Adding vertices $u_0$ and $u_2$ to any induced linear forest of $H^*$ leads to an induced linear forest of $G$. By Observation~\ref{abg} applied to $(\alpha, \beta, \gamma) = (5,12,2)$, $2 < \frac{9}{11}5 - \frac{2}{11}12$, a contradiction.

Now $v_0$ and $v_2$ are distinct. Suppose that $v_0v_2 \in E$. We have a chain \linebreak $(\{u_0,u_2\},\{u_1,u_3,v_2\},u_0,v_0)$. By Lemma~\ref{schain}, $2 < \frac{9}{11}5 - \frac{2}{11}(13 - \frac{1}{2})$, a contradiction.

Now $v_0v_2 \notin E$. Let $H^* = G - \{u_0,u_1,u_2,u_3,v_0,v_2\}$. Graph $H^*$ has $n-6$ vertices and $m-16$ edges. Adding vertices $u_0$ and $u_2$ to any induced linear forest of $H^*$ leads to an induced linear forest of $G$. By Observation~\ref{abg} applied to $(\alpha, \beta, \gamma) = (6,16,2)$, $2 < \frac{9}{11}6 - \frac{2}{11}16$, a contradiction.
\end{proof}

\begin{lemm} \label{4443}
There is no $4$-face with exactly one $3$-vertex in $G$.
\end{lemm}

\begin{proof}
By contradiction, suppose there is a $4$-face $u_0u_1u_2u_3$, such that $u_0$ is a $3$-vertex and the other $u_i$ are $4$-vertices. Let $v$ be the third neighbour of $u_0$. Note that $v$ is a $4$-vertex by Lemma~\ref{3-3}.

Suppose first that $vu_2 \in E$. By planarity of $G$, $\{u_0,v,u_2\}$ separates the vertices $u_1$ and $u_3$. Therefore $(\{u_0\}, \{v,u_2\}, u_0, u_0, u_1, u_3)$ is a separating double chain of $G$. By Lemma~\ref{sepdchain}, $1 < \frac{9}{11}3 - \frac{2}{11}(9 - 1)$, a contradiction.

Now $vu_2 \notin E$. Let $w_0$ and $w_1$ be the neighbours of $u_1$ distinct from $u_0$ and $u_2$. 

\begin{itemize}
\item Suppose $w_0$ and $w_1$ are adjacent to $u_3$.  Let $H^* = G - \{u_0,u_1,u_2,u_3,v,w_0,w_1\}$. By Lemma~\ref{degge3}, the $w_i$'s are $3^+$-vertices, and since $G$ is triangle-free, they cannot be adjacent to $u_2$. Moreover, by planarity, at least one of the $w_i$'s is not adjacent to $v$. This implies that $H^*$ has at most $m-15$ edges. Graph $H^*$ has $n-7$ vertices. Adding vertices $u_0$, $u_1$, and $u_3$ to any induced linear forest of $H^*$ leads to an induced linear forest of $G$. By Observation~\ref{abg} applied to $(\alpha, \beta, \gamma) = (7,15,3)$, $3 < \frac{9}{11}7 - \frac{2}{11}15$, a contradiction.

\item Suppose that one of the $w_i$'s, say $w_0$, is adjacent to $u_3$ and that the other one ($w_1$) is not adjacent to $u_3$.
Let $w_2$ be the neighbour of $u_3$ distinct from $u_0$, $u_2$, and $w_0$. 

\begin{itemize}
\item Suppose $w_1$ or $w_2$, say $w_1$, is a $3$-vertex in $G - \{u_0,u_1,u_2,u_3,v,w_0\}$. Suppose $w_2$ is a $3$-vertex in $G - \{u_0,u_1,u_2,u_3,v,w_0,w_1\}$ (note that this implies that $w_1w_2 \notin E$). Let $H^* = G - \{u_0,u_1,u_2,u_3,v,w_0,w_1,w_2\}$. Graph $H^*$ has $n-8$ vertices and at most $m-20$ edges. Adding vertices $u_0$, $u_1$, and $u_3$ to any induced linear forest of $H^*$ leads to an induced linear forest of $G$. By Observation~\ref{abg} applied to $(\alpha, \beta, \gamma) = (8,20,3)$, $3 < \frac{9}{11}8 - \frac{2}{11}20$, a contradiction. 

Now $w_2$ is a $2^-$-vertex in $G - \{u_0,u_1,u_2,u_3,v,w_0,w_1\}$, and thus $(\{u_0,u_1,u_3\},$ $\{u_2,v,w_0,w_1\},$ $u_3,$ $w_2)$ is a chain of $G$. By Lemma~\ref{schain}, $3 < \frac{9}{11}7 - \frac{2}{11}(17 - \frac{1}{2})$, a contradiction. 

\item Suppose $w_1$ and $w_2$ are $2^-$-vertex in $G - \{u_0,u_1,u_2,u_3,v,w_0\}$. We have a double chain $(\{u_0,u_1,u_3\},\{u_2,v,w_0\},u_1,u_3,w_1,w_2)$. By Lemma~\ref{dchain}, $3 < \frac{9}{11}6 - \frac{2}{11}(14 - 3)$, a contradiction.
\end{itemize}

\item Suppose the $w_i$'s are not adjacent to $u_3$. Let us prove by contradiction that the $w_i$'s are $2^-$-vertices in $G - \{u_0,u_1,u_2,v\}$.

 \begin{itemize}
\item Suppose the $w_i$'s are $3$-vertices in $G - \{u_0,u_1,u_2,v\}$.  Then we have the following chain: $(\{u_0,u_1\},\{u_2,v,w_0,w_1\},u_0,u_3)$. By Lemma~\ref{schain}, $2 < \frac{9}{11}6 - \frac{2}{11}(18 - \frac{1}{2})$, a contradiction. 

\item Now one of the $w_i$'s, say $w_0$, is a $2^-$-vertex in $G - \{u_0,u_1,u_2,v\}$. Suppose $w_1$ is a $3$-vertex in $G - \{u_0,u_1,u_2,v\}$. Then we have the following double chain: $(\{u_0,u_1\},\{u_2,v,w_1\},u_0,u_1,u_3,w_0)$. By Lemma~\ref{dchain}, $2 < \frac{9}{11}5 - \frac{2}{11}(15 - 3)$, a contradiction.
\end{itemize}
Now the $w_i$'s are $2^-$-vertices in $G - \{u_0,u_1,u_2,v\}$. 
The $w_i$'s are $3$-vertices or $4$-vertices in $G$, they are not adjacent to $u_0$ and $u_2$ since $G$ is triangle-free, and by Lemma~\ref{4433}, if for some $i$, $w_i$ is adjacent to $v$, then $w_i$ is a $4$-vertex. Therefore each of the $w_i$'s is either a $3$-vertex non-adjacent to $v$ or a $4$-vertex adjacent to $v$, and thus the $w_i$'s are $2$-vertices in $G - \{u_0,u_1,u_2,v\}$.

Let $x_0$ and $x_1$ be the two neighbours of $w_0$ in $G - \{u_0,u_1,u_2,u_3,v\}$. Let $d$ be the sum of the degrees of $x_0$ and $x_1$ in $G - \{u_0,u_1,u_2,v,w_0,w_1\}$. 

\begin{itemize}
\item Suppose $d \ge 4$. We have a simple chain $(\{u_0,u_1,w_0\},\{u_2,v,w_1,x_0,x_1\},u_0,u_3)$. By Lemma~\ref{schain}, $3 < \frac{9}{11}8 - \frac{2}{11}(20 - \frac{1}{2})$, a contradiction. 

\item Suppose $d \le 3$. 
Suppose one of the $x_i$, say $x_0$, is a $0$-vertex in $G - \{u_0,u_1,u_2,v,w_0,w_1\}$. We have a simple chain $(\{u_0,u_1,w_0,x_0\},\{u_2,v,w_1,x_1\},u_0,u_3)$. By Lemma~\ref{schain}, $4 < \frac{9}{11}8 - \frac{2}{11}(16 - \frac{1}{2})$, a contradiction. 

Suppose one of the $x_i$, say $x_0$, is a $1$-vertex in $G - \{u_0,u_1,u_2,v,w_0,w_1\}$ and the other one ($x_1$) is a $1$-vertex or a $2$-vertex in $G - \{u_0,u_1,u_2,v,w_0,w_1\}$. 

Let us prove by contradiction that $x_0$ is not adjacent to $u_3$. Suppose $x_0$ is adjacent to $u_3$.

Suppose $x_0w_1 \in E$. By Lemma~\ref{4433}, at least one of the $w_i$'s, $w_0$ say, is a $4$-vertex, and thus is adjacent to $v$. By planarity, $w_1$ is not adjacent to $v$, and thus $w_1$ is a $3$-vertex in $G$. In this case, $x_0$ is adjacent exactly to $w_0$, $w_1$ and $u_3$ since it is a $1$-vertex in $G - \{u_0,u_1,u_2,v,w_0,w_1\}$ and $G$ is triangle-free. Then $x_0$ and $w_1$ are adjacent $3$-vertices in $G$, which contradicts Lemma~\ref{3-3}.

Now $x_0w_1 \notin E$. We have a simple chain $(\{u_0,u_1,x_0\},\{u_2,u_3,v,w_0\},u_1,w_1)$. By Lemma~\ref{schain}, $3 < \frac{9}{11}7 - \frac{2}{11}(16 - \frac{1}{2})$, a contradiction. 

Therefore we know that $x_0$ is not adjacent to $u_3$. Let $y$ be the neighbour of $x_0$ in $G - \{u_0,u_1,u_2,v,w_0,w_1\}$. Suppose $y$ is a $3$-vertex in $G - \{u_0,u_1,u_2,v,w_0,w_1,x_0,x_1\}$. Now the following quadruplet is a simple chain: $(\{u_0,u_1,w_0,x_0\},\{u_2,v,w_1,x_1,y\},u_0,u_3)$. By Lemma~\ref{schain}, $4 < \frac{9}{11}9 - \frac{2}{11}(21 - \frac{1}{2})$, a contradiction. Now $y$ is a $2^-$-vertex in $G - \{u_0,u_1,u_2,v,w_0,w_1,x_0,x_1\}$. Then $(\{u_0,u_1,w_0,x_0\},\{u_2,v,w_1,x_1\},$ $u_0,x_0,u_3,y)$ is a double chain. By Lemma~\ref{dchain}, $4 < \frac{9}{11}8 - \frac{2}{11}(18 - 3)$, a contradiction.
\end{itemize}
\end{itemize}
\vspace{-0.8cm}
\end{proof}

For every face $f$ of $G$, let $l(f)$ denote the length of $f$, and let $c_{4}(f)$ denote the number of $4$-vertices in $f$. For every vertex $v$, let $d(v)$ be the degree of $v$. Let $k$ be the number of faces of $G$, and for every $3 \le d \le 4$ and every $4 \le l$, let $k_l$ be the number of $l$-faces and $n_d$ the number of $d$-vertices in $G$.

Each $4$-vertex is in the boundary of at most four faces. Therefore the sum of the $c_{4}(f)$ over all the $4$-faces and $5$-faces is $\sum_{f,4 \le l(f) \le 5} c_{4}(f) \le 4n_4$.
Now, by Lemmas~\ref{degge3}, \ref{4433}, and \ref{4443}, every $4$-face of $G$ has only $4$-vertices in its boundary, so for each $4$-face $f$, $c_{4}(f) = 4$. By Lemma~\ref{3-3}, every $5$-face of $G$ has at least three $4$-vertices, so for each $5$-face $f$ we have $c_{4}(f) \ge 3 \ge 2$.
Thus $\sum_{f,l(f) = 4} c_{4}(f) + \sum_{f,l(f) = 5} c_{4}(f) \ge 4k_4 + 2k_5$. Thus $4n_4\ge 4k_4 + 2k_5$, and thus $2n_4 \ge 2k_4 + k_5$.
By Euler's formula, we have:

\begin{eqnarray*} 
-12 & = & 6m - 6n - 6k  \\
& = & 2\sum_{v \in V}d(v) + \sum_{f \in F(G)}l(f) - 6n - 6k  \\
& = & \sum_{d \ge 3}(2d - 6)n_d + \sum_{l \ge 4}(l-6)k_l  \\
& \ge & 2n_4 - 2k_4 - k_5  \\
& \ge & 0 \\
\end{eqnarray*}
That contradiction ends the proof of Theorem~\ref{main}.

\bibliographystyle{plain}
\bibliography{biblio} {}

\end{document}